\documentclass{article}
\usepackage{fullpage}
\usepackage[T1]{fontenc}
\usepackage[utf8]{inputenc}
\usepackage[all=normal,bibliography=tight]{savetrees}

  \usepackage{amstext,amsfonts,amsthm,amsmath,amssymb}

\usepackage{graphicx,tikz}
\usetikzlibrary{snakes,shapes}
\usetikzlibrary{calc}
\usetikzlibrary{decorations.markings}
\usetikzlibrary{patterns}
\usepackage{comment}
\usepackage{url}
\usepackage{xspace}
\usepackage{todonotes}
\usepackage[shortlabels]{enumitem}
\usepackage[absolute]{textpos}
\usepackage{cleveref}
\usepackage{wrapfig}

\newcommand{\tw}{\mathrm{tw}}
\newcommand{\td}{\mathrm{td}}
\newcommand{\treewidth}{\mathrm{tw}}
\newcommand{\treedepth}{\mathrm{td}}

\newcommand{\Oh}{\mathcal{O}}
\newcommand{\col}{\alpha}
\newenvironment{claimproof}{\begin{proof}}{\cqed\end{proof}}

\def\cqedsymbol{\ifmmode$\lrcorner$\else{\unskip\nobreak\hfil
\penalty50\hskip1em\null\nobreak\hfil$\lrcorner$
\parfillskip=0pt\finalhyphendemerits=0\endgraf}\fi} 

\newcommand{\cqed}{\renewcommand{\qed}{\cqedsymbol}}

\newtheorem{lemma}{Lemma}[section]
\newtheorem{proposition}[lemma]{Proposition}

\newtheorem{corollary}[lemma]{Corollary}
\newtheorem{theorem}[lemma]{Theorem}
\newtheorem{claim}[lemma]{Claim}
\theoremstyle{definition}
\newtheorem{definition}[lemma]{Definition}

\crefname{lemma}{Lemma}{Lemmas}
\crefname{theorem}{Theorem}{Theorems}

\title{Improved bounds for the excluded-minor approximation of treedepth\thanks{This research is a part of a project that have received funding from the European Research Council (ERC)
under the European Union's Horizon 2020 research and innovation programme
Grant Agreement 714704. An extented abstract of this work appeared at ESA 2019~\cite{esa-version}.}}

\author{ 
 Wojciech Czerwi\'{n}ski\footnote{University of Warsaw, Poland, \texttt{W.Czerwinski@mimuw.edu.pl}}
   \and Wojciech Nadara\footnote{University of Warsaw, Poland, \texttt{W.Nadara@mimuw.edu.pl}}
     \and Marcin Pilipczuk\footnote{University of Warsaw, Poland, \texttt{M.Pilipczuk@mimuw.edu.pl}}
}

\date{}

\begin{document}

\maketitle

\begin{textblock}{20}(0, 12.4)
\includegraphics[width=40px]{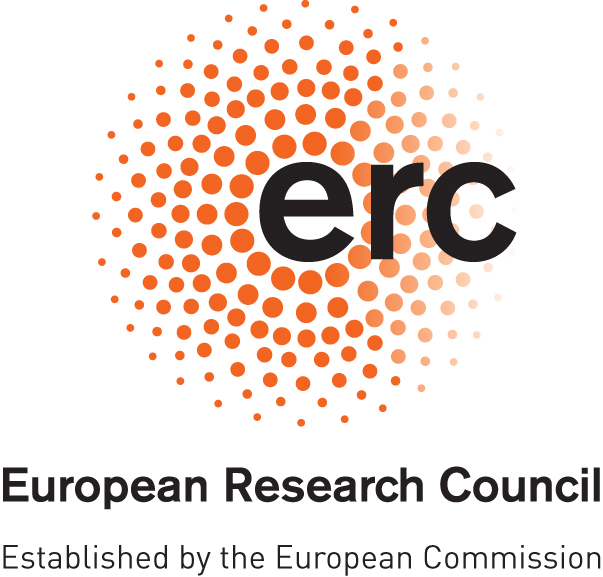}%
\end{textblock}
\begin{textblock}{20}(0, 13.4)
\includegraphics[width=40px]{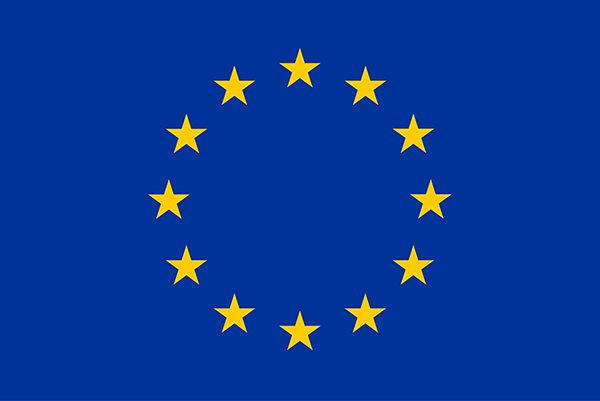}%
\end{textblock}

\begin{abstract}
Treedepth, a more restrictive graph width parameter than treewidth and pathwidth,
plays a major role in the theory of sparse graph classes. 
We show that there exists a constant $C$ such that for every positive integers $a,b$ and a graph $G$,
if the treedepth of $G$ is at least $Cab$, then the treewidth of $G$ is at least $a$
or $G$ contains a subcubic (i.e., of maximum degree at most $3$) tree of treedepth at least $b$ as a subgraph. 

As a direct corollary, we obtain that every graph of treedepth $\Omega(k^3)$ is either of treewidth at least $k$,
contains a subdivision of full binary tree of depth $k$, or contains a path of length $2^k$.
This improves the bound of $\Omega(k^5 \log^2 k)$ of Kawarabayashi and Rossman [SODA 2018].

We also show an application of our techniques for approximation algorithms of treedepth: given a graph $G$ of treedepth $k$
and treewidth $t$, one can in polynomial time compute a treedepth decomposition of $G$ of width
$\Oh(kt \log^{3/2} t)$. This improves upon a bound of $\Oh(kt^2 \log t)$ stemming from a tradeoff between known results.

The main technical ingredient in our result is a proof that every tree of treedepth $d$
contains a subcubic subtree of treedepth at least $d \cdot \log_3 ((1+\sqrt{5})/2)$.  
\end{abstract}

\section{Introduction}
For an undirected graph $G$, the \emph{treedepth} of $G$ is the minimum height of a rooted forest
whose ancestor-descendant closure contains $G$ as a subgraph. 
Together with more widely known related width notions such as treewidth and pathwidth, 
it plays a major role in structural graph theory, in particular in the study of general sparse graph classes~\cite{NesetrilM15,sparsity-book,NesetrilM06}.

An important property of treedepth is that it admits a number of equivalent definitions. 
Following the definition of treedepth above, a \emph{treedepth decomposition} of a graph $G$ consists 
of a rooted forest $F$ and an injective mapping $f : V(G) \to V(F)$ such that for every $uv \in E(G)$ the vertices
$f(u)$ and $f(v)$ are in ancestor-descendant relation in $F$. The \emph{width} of a treedepth decomposition $(F,f)$ is the height
of $F$ (the number of vertices on the longest leaf-to-root path in $F$) and the treedepth of $G$ is the minimum possible height of a treedepth decomposition of $G$.
A \emph{centered coloring} of a graph $G$ is an assignment $\col : V(G) \to \mathbb{Z}$ such that
for every connected subgraph $H$ of $G$, $\col$ has a \emph{center} in $H$: a vertex $v \in V(H)$ of unique color, 
i.e., $\col(v) \neq \col(w)$ for every $w \in V(H) \setminus \{v\}$. 
A \emph{vertex ranking} of a graph $G$ is an assignment $\col : V(G) \to \mathbb{Z}$ such that in
every connected subgraph $H$ of $G$ there is a unique vertex of \emph{maximum} rank (value $\col(v)$).
Clearly, each vertex ranking is a centered coloring.
It turns out that the minimum number of colors (minimum size of the image of $\col$) needed for a centered coloring and for a vertex ranking are equal and equal to the treedepth of a graph~\cite{NesetrilM06}.

While there are multiple examples of algorithmic usage of treedepth in the theory of sparse graphs~\cite{sparsity-book,NesetrilM15},
our understanding of the complexity of computing minimum width treedepth decompositions is limited. 
For a graph $G$, let $\td(G)$ and $\tw(G)$ denote the treedepth and the treewidth of $G$, respectively.
An algorithm of Reidl, Rossmanith, Villaamil, and Sikdar~\cite{ReidlRVS14} computes exactly the treedepth
of an input graph $G$ in time $2^{\Oh(\td(G) \cdot t)} n^{\Oh(1)}$, given a tree decomposition of $G$
of width $t$. Combined with the classic
constant-factor approximation algorithm for treewidth that runs in $2^{\Oh(\tw(G))} n^{\Oh(1)}$
time~\cite{RobertsonS86}, one obtains an exact algorithm for treedepth running in time
$2^{\Oh(\td(G) \tw(G))} n^{\Oh(1)}$. 
No faster exact algorithm is known.

For approximation algorithms, the following folklore lemma (presented with full proof in~\cite{KR18}) is very useful.
\begin{lemma}\label{lem:tw2td}
Given a graph $G$ and a tree decomposition $(T,\beta)$ of $G$ of maximum bag size $w$, one
can in polynomial time compute a treedepth decomposition of $G$ of width at most $w \cdot \treedepth(T)$.
\end{lemma}
Using Lemma~\ref{lem:tw2td}, one can obtain an approximation algorithm for treedepth 
with a cheap tradeoff trick.%
\footnote{This trick has been observed and communicated to us by Micha\l{} Pilipczuk.
  We thank Micha\l{} for allowing us to include it in this paper.}
\begin{lemma}\label{lem:td-cheap-apx}
Given a graph $G$, one can in polynomial time compute a treedepth decomposition of $G$
of width $\Oh(\treedepth(G) \cdot \treewidth(G)^2 \log \treewidth(G))$. 
\end{lemma}
\begin{proof}
Let $n = |V(G)|$.
Using the polynomial-time approximation algorithm for treewidth~\cite{FeigeHL08}, compute
a tree decomposition $(T,\beta)$ of $G$ of width $t = \Oh(\treewidth(G) \sqrt{\log \treewidth(G)})$
and $\Oh(n)$ bags.
For every integer $1 \leq k \leq (\log n) / t$, use the algorithm of~\cite{ReidlRVS14}
to check in polynomial time if the treedepth of $G$ is at most $k$.
Note that if this is the case, the algorithm finds an optimal treedepth decomposition and we can conclude.
Otherwise, we have $\log n \leq \td(G) \cdot t$ and we 
apply Lemma~\ref{lem:tw2td} to $G$ and $(T,\beta)$ obtaining a treedepth decomposition 
of $G$ of width
\[ \Oh(t \log n) \leq \Oh(\treedepth(G) \cdot t^2) \leq \Oh( \treedepth(G) \cdot \treewidth(G)^2 \log \treewidth(G)). \]
\end{proof}
Lemma~\ref{lem:td-cheap-apx} is the only polynomial approximation algorithm
for treedepth running in polynomial time we are aware of.

A related topic to exact and approximation algorithms computing minimum-width treedepth decomposition
is the study of obstructions to small treedepth. Dvo\v{r}\'{a}k, Giannopoulou, and Thilikos~\cite{DvorakGT12}
proved that every minimal graph of treedepth $k$ has the number of vertices at most double-exponential in $k$.
More recently, Kawarabayashi and Rossman showed an excluded-minor theorem for treedepth.
\begin{theorem}[\cite{KR18}]
There exists a universal constant $C$ such that for every integer $k$ every graph of treedepth 
at least $Ck^5 \log^2k$ is either of treewidth at least $k$, contains a subdivision of a full binary tree of depth $k$
as a subgraph, or contains a path of length $2^k$.
\end{theorem}
While neither the results of~\cite{DvorakGT12} nor~\cite{KR18} have a direct application in the approximability of treedepth,
these topics are tightly linked with each other and we expect that a finer understanding of treedepth obstructions
is necessary to provide more efficient algorithms computing or approximating the treedepth of a graph.

\paragraph*{Our results.}
Our main graph-theoretical result is the following statement, improving upon the work of Kawarabayashi and Rossman~\cite{KR18}.

\begin{theorem}\label{thm:main}
Let $G$ be a graph of treewidth $\tw(G)$ and treedepth $\td(G)$.
Then there exists a subcubic tree $H$ that is a subgraph of $G$ and is of
treedepth at least 
\[ \frac{\td(G)}{\tw(G) + 1} \cdot \frac{\log((1+\sqrt{5})/2)}{\log(3)}. \]
\end{theorem}
In other words, Theorem~\ref{thm:main} states that there exists a constant $C = \frac{\log(3)}{\log((1+\sqrt{5})/2)}$
such that for every graph $G$ and positive integers $a,b$,
   if the treedepth of $G$ is at least $Cab$, then the treewidth of $G$ is at least $a$ or $G$ contains a subcubic tree of treedepth $b$.
Since every subcubic tree of treedepth $d$ contains either a simple path of length $2^{\Omega(\sqrt{d})}$ 
or a subdivision of a full binary tree of depth $\Omega(\sqrt{d})$~\cite{KR18}, we have the following corollary.
\begin{corollary}\label{cor:main}
Let $G$ be a graph of treewidth $\tw(G)$ and treedepth $\td(G)$.
Then for some 
\[ h = \Omega\left(\sqrt{\td(G) / (\tw(G) + 1)}\right) \]
$G$ contains either a simple path of length $2^h$ or 
a subdivision of a full binary tree of depth $h$.

Consequently, there exists an absolute constant $C$ such that for every integer $k \geq 1$
and a graph $G$ of treedepth at least $C k^3$, either 
\begin{itemize}
\item $G$ has treewidth at least $k$, 
\item $G$ contains a subdivision of a full binary tree of depth $k$ as a subgraph, or
\item $G$ contains a path of length $2^k$.
\end{itemize}
\end{corollary}
In other words, Corollary~\ref{cor:main} improves the bound $k^5 \log^2k$ of Kawarabayashi and Rossman~\cite{KR18} to $k^3$. 
We remark here that there are subcubic trees of treedepth $\Omega(h^2)$ that contains neither a path of length $2^h$ nor a subdivision
of a full binary tree of depth $h$,%
\footnote{It is straightforward to deduce such an example from the proof of~\cite{KR18}. We provide such an example in Section~\ref{sec:ex}.}
and thus the quadratic loss between the statements of Theorem~\ref{thm:main}
and Corollary~\ref{cor:main} is necessary.

Inside the proof of Theorem~\ref{thm:main} we make use of the following lemma that may be of independent interest.
This lemma is the main technical improvement upon the work of Kawarabayashi and Rossman~\cite{KR18}.

\begin{lemma}\label{lem:trees}
Every tree of treedepth $d$ contains a subcubic subtree of treedepth at least $\frac{\log((1+\sqrt{5})/2)}{\log(3)} d$.
\newline Furthermore, such a subtree can be found in polynomial time.
\end{lemma}

Lemma~\ref{lem:trees}, developed to prove Theorem~\ref{thm:main}, have some implications on the approximability of treedepth.
We combine it with the machinery of Kawarabayashi and Rossman~\cite{KR18} to improve upon Lemma~\ref{lem:td-cheap-apx} as follows.
\begin{theorem}\label{thm:apx}
Given a graph $G$, one can in polynomial time compute a treedepth decomposition of $G$
of width $\Oh(\treedepth(G) \cdot \treewidth(G) \log^{3/2} \treewidth(G))$. 
\end{theorem}

The result of Kawarabayashi and Rossman~\cite{KR18} has been also an important ingredient in the study of
\emph{linear colorings}~\cite{LinCol}. A coloring $\col : V(G) \to \mathbb{Z}$ of a graph $G$
is a \emph{linear coloring} if for every (not necessarily induced) path $P$ in $G$ there exists 
a vertex $v \in V(P)$ of unique color $\col(v)$ on $P$. Clearly, each centered coloring is a linear coloring,
but the minimum number of colors needed for a linear coloring can be much smaller than the treedepth of a graph. 
Kun et al.~\cite{LinCol} provided a polynomial relation between the treedepth and the minimum number of colors
in a linear coloring; by replacing their usage of~\cite{KR18} by our result 
(and using an improved bound for the excluded grid theorem~\cite{ChuzhoyT19}) we obtain an improved bound.
\begin{theorem}\label{thm:lincol}
There exists a polynomial $p$ such that for every integer $k$ and graph $G$, if the treedepth of $G$
is at least $k^{19} p(\log k)$, then every linear coloring of $G$ requires at least $k$ colors.
\end{theorem}
The previous bound of~\cite{LinCol} is $k^{190} p(\log k)$. 

\medskip

Brief preliminaries on tree decompositions, treewidth, and brambles can be found in Section~\ref{sec:prelims}.
After proving Lemma~\ref{lem:trees} in Section~\ref{sec:trees}, 
we prove Theorem~\ref{thm:main} in Section~\ref{sec:simple}.

Then in Section~\ref{sec:kr} we provide a proof a lemma
 that combines the machinery of Kawarabayashi and Rossman~\cite{KR18}
with Lemma~\ref{lem:trees} and conclude the proof of Theorem~\ref{thm:apx} using it.
Finally, Theorem~\ref{thm:lincol} is proven in Section~\ref{sec:lincol}.

\section{Preliminaries}\label{sec:prelims}
The symbol $\log_p$ stands for base-$p$ logarithm and $\log$ stands for $\log_2$.
We denote $\varphi = \frac{1 + \sqrt{5}}{2}$; note that $\varphi$ is chosen in a way so that $\varphi^2 = \varphi + 1$ and $\varphi > 1$.

We need a few basic notions concerning tree decompositions and treewidth.
Recall that a \emph{tree decomposition} of a graph $G$ is a pair $(T,\beta)$ where $T$ is a rooted tree and $\beta : V(T) \to 2^{V(G)}$
is such that for every $v \in V(G)$ the set $\{t \in V(T)~|~v \in \beta(t)\}$ induces a connected nonempty subtree of $T$
and for every $uv \in E(G)$ there exists $t \in V(T)$ with $u,v \in \beta(t)$. 
The width of a tree decomposition $(T,\beta)$ is $\max_{t \in V(T)} |\beta(t)|-1$ and the treewidth of a graph
is the minimum possible width of its tree decomposition.

A \emph{bramble} in a graph $G$ is a family $\mathcal{B}$ of connected subgraphs of $G$ such that for every
$B_1,B_2 \in \mathcal{B}$, either $B_1$ and $B_2$ share a vertex or there is an edge of $G$ with one endpoint in $B_1$
and one endpoint in $B_2$. Standard arguments (see e.g.~\cite{diestel}) show the following:
\begin{lemma}\label{lem:hit-bramble}
Let $G$ be a graph, $(T,\beta)$ be a tree decomposition of $G$, and let $\mathcal{B}$ be a bramble in $G$.
Then there exists $t \in V(T)$ such that for every $B \in \mathcal{B}$ it holds that $\beta(t) \cap V(B) \neq \emptyset$.
\end{lemma}

\section{Subcubic subtrees of trees of large treedepth}\label{sec:trees}
This section focuses on proving Lemma~\ref{lem:trees}.

Sch\"{a}ffer~\cite{Schaffer89} proved that there is a linear time algorithm for finding a vertex ranking with minimum number
of colors of a~tree $T$.
We follow \cite{LinCol} for a good description of its properties.

In original Sch\"{a}ffer’s algorithm ranks are starting from $1$, however
for the ease of exposition let us assume that ranks are starting from $0$.
That is, the algorithm constructs a vertex ranking $\col : V(T) \to \{0,1,2,\ldots\}$
trying to minimize the maximum value attained by $\col$.
Assume that $T$ is rooted in an arbitrary vertex and for every $v \in V(T)$ let $T_v$ be the subtree rooted at $v$.

Of central importance to Sch\"{a}ffer’s algorithm are what we will refer to as \textit{rank lists}.
For a rooted tree $T$, the rank list $L(T)$ for vertex ranking $\col$ consists of these ranks $i$ for which there exists a path $P$ starting from the
root and ending in a vertex $v$ with $\col(v) = i$ such that for every $u \in V(P) \setminus \{v\}$ we have $\col(u) < \col(v)$, that is,
$v$ is the unique vertex of maximum rank on $P$. More formally:

\begin{definition} For a vertex ranking $\col$ of tree $T$, the rank list of $T$, denoted
$L(T)$, can be defined recursively as $L(T) = L(T \setminus T_v) \cup \{\col(v)\}$ where $v$ is the
vertex of maximum rank in $T$.
\end{definition}

Sch\"{a}ffer’s algorithm arbitrarily roots $T$ and builds the ranking from the leaves
to the root of $T$, computing the rank of each vertex from the rank lists of each
of its children. For brevity, we denote $L(v) = L(T_v)$ for every $v$ in $T$.

\begin{proposition} Let $\col$ be a~vertex ranking of $T$ produced by Sch\"{a}ffer’s
algorithm and let $v \in T$ be a~vertex with children $u_1, \ldots, u_l$. If $x$ is the largest
integer appearing on rank lists of at least two children of $v$ (or~$-1$ if all such rank
lists are pairwise disjoint) then $\col(v)$ is the smallest integer satisfying $\col(v) > x$
and $\col(v) \not\in \bigcup_{i=1}^{l} L(u_i)$.
\end{proposition}

For a node $v \in V(T)$, and vertex ranking $\col$, the following potential
is pivotal to the analysis of Sch\"{a}ffer's algorithm.
Let $l_0 > l_1 > \ldots > l_{|L(v)| - 1}$ be the elements of $L(v)$ sorted in decreasing order.
\[ \zeta(v) = \sum_{r \in L(v)} 3^r = \sum_{i=0}^{|L(v)|-1} 3^{l_i}. \]
When we write $\zeta(T)$ for some tree $T$ we refer to $\zeta(s)$
where $s$ is a root of $T$.
For our purposes, we will also use a skewed version of potential function with a different base
\[ \sigma(v) = \sum_{i=0}^{|L(v)|-1} \varphi^{l_i - i}, \]
where again $l_0 > l_1 > \ldots > l_{|L(v)| - 1}$ are elements of $L(v)$ sorted in decreasing order.
Throughout this section, when focusing on one node $v \in V(T)$, we use notation that $l_i$ is $i-$th element of set $L(v)$ when sorted in decreasing order
and when $0-$based indexed.

Let us start with proving following two bounds that estimate $\td(T)$ in terms of $\zeta(T)$ and $\sigma(T)$.

\begin{claim} \label{cl:bd1} $\log_{\varphi}(\sigma(T)) \ge \td(T) - 1$.
\end{claim}

\begin{claimproof}
We know that $L(T)$ is nonempty and its biggest element is equal to $\td(T)-1$ (we need to subtract one because
we use nonnegative numbers as ranks, not positive). Therefore we have
\[ \sigma(T) = \sum_{i=0}^{|L(T)|-1} \varphi^{l_i-i} \ge \varphi^{l_0} = \varphi^{\td(T) - 1}. \]
Hence, $\log_{\varphi}(\sigma(T)) \ge \td(T) - 1$, as desired.
\end{claimproof}

\begin{claim} \label{cl:bd2}  $\log_3(\zeta(T)) + \log_3(2) < \td(T) $.
\end{claim}

\begin{claimproof}
We have that 
\begin{align*} 
\zeta(T) = \sum_{r \in L(v)} 3^r &\le \sum_{r=0}^{\td(T)-1}3^r = \frac{3^{\td(T)} - 1}{2}, \\
2 \zeta(T) &\le 3^{\td(T)} - 1 < 3^{\td(T)} \\
\log_3(2) + \log_3(\zeta(T)) &< \td(T).
\end{align*} 
\end{claimproof}

We are ready to prove Lemma \ref{lem:trees}.
Given tree $T$ we want to produce a subcubic (i.e., maximum degree at most $3$) tree $S$ which is a subtree of $T$
and that fulfills $\td(S) > \td(T) \log_3(\varphi)$.

Let us start our algorithm by arbitrary rooting $T$ and computing rank lists using Sch\"{a}ffer’s algorithm.
Then for every vertex $v \in T$ we define $C(v)$ as a set of two children of $v$
that have the biggest value of $\zeta$ in case $v$ has at least two children, or all children otherwise.
Let us now define forest $F$ whose vertex set is the same as vertex set of $T$
where for every $v$ we put edges between $v$ and all elements of $C(v)$.
Clearly this is a forest consisting of subcubic trees which are subtrees of $T$ (where \textit{subtree} is understood as subgraph, not necessarily as some vertex $t$ along with all its descendants in a rooted tree).
Let $S$ be a tree of this forest containing root of $T$.
We claim that $S$ is that subcubic subtree of $T$ we are looking for.
Note that computing $F$ and thus $S$ can be trivially done in polynomial time.
Hence, we are left with proving that $\td(S) > \td(T) \log_3(\varphi)$.

Let us root every tree of $F$ in a vertex that was closest to root of $T$ in $T$.
Then compute rank lists for these trees using Sch\"{a}ffer’s algorithm. So now,
for every vertex we have two rank lists, one for $T$ and one for $F$.
Let us now denote these second ranklists as $\widetilde{L}(v)$ for $v \in V(T)$
and let us define function $\widetilde{\zeta}$ which will be similar potential function as $\zeta$, but operating
on rank lists $\widetilde{L}(v)$ instead of $L(v)$.
Following claim will be crucial.
\begin{claim} \label{claim:ineq_pots} 
For every $v \in V(T)$ it holds that $\widetilde{\zeta}(v) \ge \sigma(v)$.
\end{claim}
We first verify that Claim~\ref{claim:ineq_pots} implies Lemma~\ref{lem:trees}.
\begin{proof}[Proof of Lemma~\ref{lem:trees}.]
Using also Claims~\ref{cl:bd2} and~\ref{cl:bd1}
we infer that
\begin{align*}
\td(S) &> \log_3(\widetilde{\zeta}(S)) + \log_3(2) & \textrm{by Claim~\ref{cl:bd2} for $S$}\\ 
      &\ge \log_3(\sigma(T)) + \log_3(2) & \textrm{by Claim~\ref{claim:ineq_pots}}\\
    &= \log_{\varphi}(\sigma(T)) \cdot \log_3(\varphi) + \log_3(2) & \textrm{logarithm base change}\\
    &\ge (\td(T) - 1) \cdot \log_3(\varphi) + \log_3(2) & \textrm{by Claim~\ref{cl:bd1} for $T$}\\
  &= \td(T) \cdot \log_3(\varphi) - \log_3(\varphi) + \log_3(2) > \td(T) \cdot \log_3(\varphi). &
\end{align*}
\end{proof}

Thus it remains to prove Claim~\ref{claim:ineq_pots}.
To this end, we prove two auxiliary inequalities.
\begin{claim} \label{lem:pot3}
For every $v \in V(T)$ it holds that $\widetilde{\zeta}(v) \ge 1 + \sum_{s \in C(v)} \widetilde{\zeta}(s)$
\end{claim}

\begin{claimproof}
We express every $\widetilde{\zeta}(x)$ for $x \in \{v\} \cup C(v)$
as a sum of powers of $3$ and count how many times each power occurs on both sides of this claimed inequality. 
Consider a summand $3^c$. If $c > \col(v)$ then, by the choice of $\col(v)$,
         $3^c$ appears at most once on the right side and if it appears there, then it appears on the left side as well, so contributions of summands of form $3^c$ for $c > \col(v)$ to both sides are equal. The summand $3^{\col(v)}$ appears once on the left side and does not appear on the right side.
         For $c < \col(v)$, the summands of form $3^c$ appear at most twice in $\sum_{s \in C(v)} \widetilde{\zeta}(s)$,
so their contribution to right side can be bounded from above by
$\sum_{c=0}^{\col(v)-1} 2 \cdot 3^c = 3^{\col(v)}-1$, so in fact $3^{\col(v)}$ from the left side contributes at least as much as remaining summands from the right side.
This finishes the proof of the claim.
\end{claimproof}
\begin{claim} \label{lem:potphi}
For every $v \in V(T)$ it holds that $\sigma(v) \le 1 + \sum_{s \in C(v)} \sigma(s)$
\end{claim}
\begin{claimproof}
Recall that by the definition $C(v)$ is a set of two children of $v$ in $T$ with the biggest values of
$\zeta$ or a set of all children of $v$ in case it has less than two of them.
Observe that having bigger value of $\zeta(v)$ is another way of expressing having the set $L(v)$ bigger lexicographically when sorted
in decreasing order.

If $v$ is a leaf then $C(v)$ is empty and $\sigma(v) = 1$, so the inequality is obvious.
Henceforth we focus on a vertex $v$ that is not a leaf.
In our proof following equation will come handy:

\[ \varphi = \sum_{i=0}^{\infty} \varphi^{-2i} \]
It holds since $\sum_{i=0}^{\infty} \varphi^{-2i} = \frac{1}{1 - \varphi^{-2}} = \frac{\varphi^2}{\varphi^2 - 1} = \frac{\varphi^2}{\varphi} = \varphi$.

Let us now analyze $L(v)$. It consists of some prefix $P$ of values that appeared exactly once in children of~$v$, then $\col(v)$ and then nothing (when enumerated from the biggest to the smallest). Let us now denote by $A_i$ intersection of $L(u_i)$ and $P$, where
$u_i$ is $i-$th child of $v$ when sorted in nonincreasing order by their values $\zeta(u_i)$ (1-based).
We distinguish two cases:

\paragraph*{Case 1: $A_2$ is nonempty.}
If $A_2$ is nonempty then in particular it means that $v$ has at least two children.
Let us denote the biggest element of $L(u_2)$ by $d$. We have that $d \in P$, but $d$ is not the biggest element of $P$.
Its contribution to $\sigma(u_2)$ is $\varphi^d$,
however its contribution to $\sigma(v)$ is at most $\varphi^{d-1}$ (because of the skew and since $d$ is not the biggest element of $P$).
Contribution to $\sigma(v)$ of elements smaller than $d$ can be bounded from above by $\varphi^{d-3} + \varphi^{d-5} + \ldots$.
We know that $d = l_j$ for some $j$, where $j \ge 1$ and $L(v)$ consists of elements $l_0 >l_1 > \ldots > l_{|L(v)-1|}$.
We have that $l_k \in L(u_1)$ for $k<j$ and that $l_j \ge l_i + (i-j)$~for~$i \ge j$, so $l_i - i \le l_j - j - 2(i-j)= d-j -2(i-j)$.

We can deduce that 
\begin{align*}
\sigma(v) &= \sum_{i=0}^{|L(v)| - 1} \varphi^{l_i - i} = \sum_{i=0}^{j - 1} \varphi^{l_i - i} + \sum_{i=j}^{|L(v)| - 1} \varphi^{l_i - i} \le
\sigma(u_1) + \sum_{i=j}^{|L(v)| - 1} \varphi^{d-j-2(i-j)} \\
&\le \sigma(u_1) + \varphi^{d-j} \sum_{i=j}^{\infty} \varphi^{-2(i-j)} =
\sigma(u_1) + \varphi^{d-j} \sum_{i=0}^{\infty} \varphi^{-2i} = \sigma(u_1) + \varphi^{d-j+1} \\
    &\le \sigma(u_1) + \varphi^d \le \sigma(u_1) + \sigma(u_2) <
1 + \sigma(u_1) + \sigma(u_2),
  \end{align*}
which is what we wanted to prove.

\paragraph*{Case 2: $A_2$ is empty.}
Let us now introduce a few variables:
\begin{itemize}
\item $d$ - the biggest integer number smaller than $\col(v)$ that is not an element of $L(u_1)$. \newline We know that elements from $d+1$ to $\col(v)-1$
belong to $L(u_1)$.
\item $k$ - shorthand for number of these elements (which is equal to $\col(v)-1-d$).
\newline $k$~can be zero, but cannot be negative.
\item $g$ - the number of elements of $L(v)$ that are bigger than $\col(v)$.
\end{itemize}
Then from the definition of $\col(v)$ either
\begin{itemize}
\item $d=-1$; or
\item $v$ has at least two children and $L(u_2)$ contains a number that is at least $d$.
\end{itemize}
Because of that we have $1 + \sum_{s \in C(v)} \sigma(s) \ge \sigma(u_1) + \varphi^d$.
We know that $\sum_{s \in C(v)}$ is either $\sigma(u_1)$ or $\sigma(u_1) + \sigma(u_2)$, depending on whether $v$ has only one child or more. If $d=-1$ then
$1 \ge \varphi^{d}$ and stated inequality holds. If $d \neq -1$ then $u_2$ exists
and $\sigma(u_2) \ge \varphi^d$.

Note that either $k>0$ or $g>0$, because if $k=g=0$ then $d=\col(v)-1$ and $L(u_1)$ cannot contain elements bigger then $\col(v)$ (because $g=0$),
cannot contain $\col(v)$ (from the definition of $\col(v)$) and cannot contain $\col(v)-1$ (since $d=\col(v)-1$), so its biggest element is at most $d-1$. If $d=-1$ then it means that $v$ is a leaf, but we already assumed it is not one. However, if $v$ has at least two children and $L(u_2)$ contains a number that is at least $d$, then it contradicts the assumption that $\zeta(u_1) \ge \zeta(u_2)$. So indeed it holds that $k>0$ or $g>0$ and therefore $k + g \ge 1$.

We have that \[ \sigma(v) - \sigma(u_1) \le \varphi^{\col(v)-g} - (\varphi^{\col(v)-g-1} + \varphi^{\col(v)-g-3} + \ldots + \varphi^{\col(v)-g-2k+1}), \]
which is because summands coming from numbers bigger than $\col(v)$ in $L(v)$ and $L(u_1)$ cancel out ($A_2$~is~empty, so all elements of $L(v)$ different than $\col(v)$ come from $L(u_1)$) and new rank $\col(v)$ contributes $\varphi^{\col(v)-g}$ to $\sigma(v)$
whereas $L(u_1)$ contains numbers from $d+1$ up to $\col(v)-1$ and their contribution to $\sigma(u_1)$ is $\varphi^{\col(v)-g-1} + \varphi^{\col(v)-g-3} + \ldots + \varphi^{\col(v)-g-2k+1}$.
\newline We conclude that $\sigma(v) - \sigma(u_1) \le \varphi^{-g} (\varphi^{\col(v)} - (\varphi^{\col(v)-1} + \varphi^{\col(v)-3} + \ldots + \varphi^{\col(v)-2k+1}))$.
\newline On the other hand since $\varphi^2 = \varphi + 1$ we have that
\[ \varphi^{\col(v)} = \varphi^{\col(v)-1}+\varphi^{\col(v)-2} = \varphi^{\col(v)-1}+\varphi^{\col(v)-3} + \varphi^{\col(v)-4} = \ldots = \]
\[ = (\varphi^{\col(v)-1}+\varphi^{\col(v)-3} + \ldots + \varphi^{\col(v)-2k+1}) 
+ \varphi^{\col(v)-2k}. \]
Because of that we have
\[ \sigma(v) - \sigma(u_1) \le \varphi^{-g} \cdot \varphi^{\col(v)-2k} = \varphi^{\col(v)-2k-g} = \varphi^{\col(v)-(\col(v)-1-d)-k-g} = \varphi^{d+1-(k+g)} \le \varphi^{d}. \] From that we conclude that $\sigma(v) \le \sigma(u_1) + \varphi^d \le 1 + \sum_{s \in C(v)} \sigma(s)$, what concludes proof of this claim.
\end{claimproof}

Now, having claims \ref{lem:potphi} and \ref{lem:pot3} proven, we can wrap our reasoning up. If $v$ is a leaf then $\sigma(v) = \widetilde{\zeta}(v)=1$.
If $v$ is not a leaf then we know that $\sigma(v) \le 1 + \sum_{s \in C(v)} \sigma(s)$ and $\widetilde{\zeta}(v) \ge 1 + \sum_{s \in C(v)} \widetilde{\zeta}(s)$,
so by straightforward induction we get that $\sigma(v) \le \widetilde{\zeta}(v)$ for every $v \in V(T)$, as desired by Claim~\ref{claim:ineq_pots}.

\section{Proof of Theorem~\ref{thm:main}}\label{sec:simple}
Theorem~\ref{thm:main} is a direct corollary of Lemma~\ref{lem:trees} and the following statement.

\begin{theorem}\label{thm:simple}
Let $G$ be a graph and $a,b$ be positive integers. 
If the treewidth of $G$ is less than $a$ and $G$ does not contain any tree
of treedepth more than $b$ as a subgraph, then the treedepth of $G$ is at most $ab$.
\end{theorem}
\begin{proof}
We prove the fact by induction on $b$. 
For $b=1$, if $G$ contains no tree of treedepth $2$ as a subgraph, then $G$ is edgeless, and its treedepth is at most $1$, as desired.

Assume now $b>1$ and that the statement is true for all $b' < b$.
Without loss of generality assume that $G$ is connected, as otherwise we prove the treedepth bound for each connected
component of $G$ separatedly. 

Let $\mathcal{B}$ be the family of all subgraphs of $G$ that are trees of treedepth exactly $b$.
The crucial observation is the following.
\begin{claim}\label{cl:treesmeet}
For every two $B_1,B_2 \in \mathcal{B}$, $V(B_1) \cap V(B_2) \neq \emptyset$.
\end{claim}
\begin{proof}
Assume the contrary and let $B_1,B_2 \in \mathcal{B}$ be two offending trees in $\mathcal{B}$.
Since $b > 1$, $B_1$ and $B_2$ are nonempty. Let $P$ be a shortest path from $V(B_1)$ and $V(B_2)$; it exists as $G$ is connected.
Define a subgraph $B$ of $G$ being the union of $B_1$, $P$, and $B_2$. 
Since $P$ is a shortest path from $B_1$ to $B_2$, $B$ is a tree.
However, since the treedepth of $B_1$ and $B_2$ equals $b$, the treedepth of $B$ is more than $b$, a contradiction.
\cqed\end{proof}
Claim~\ref{cl:treesmeet} implies that $\mathcal{B}$ is a bramble in $G$.

Let $(T,\beta)$ be a tree decomposition of $G$ of minimum width.
By Lemma~\ref{lem:hit-bramble}, there exists a bag $\beta(t)$ for some $t \in V(T)$ that intersects $V(B)$ for every $B \in \mathcal{B}$.
By the definition of $\mathcal{B}$, every tree that is a subgraph of $G' := G-\beta(t)$ has treedepth less than $b$. 
By the inductive hypothesis, the treedepth of $G'$ is at most $a \cdot (b-1)$. 
Thus, the treedepth of $G$ is at most $\treedepth(G') + |\beta(t)| \leq a \cdot (b-1) + a = ab$, as desired.
\end{proof}

\section{Proof of  Theorem~\ref{thm:apx}}\label{sec:apx}\label{sec:kr}

We consider a greedy tree decomposition of a connected graph $G$, as defined in~\cite{KR18}.
A greedy tree decomposition is a tree decomposition that can be also interpreted as a treedepth decomposition.
More formally, a tree decomposition $(T, \beta)$ of a graph $G$ is \emph{greedy}
if
\begin{enumerate}
  \item $V(T) = V(G)$,
  \item for every $uv \in E(G)$, the nodes $u$ and $v$ in $T$ are in ancestor-descendant relation in $T$, and
  \item for every vertex $u \in V(T)$ and its child $v$ there is some descendant $w$ of $v$ in $T$
  such that $uw \in E(G)$.
\end{enumerate}

We now prove the following lemma that combines Lemma~\ref{lem:trees} with the machinery
of Kawarabayashi and Rossmann~\cite{KR18}.
\begin{lemma}\label{lem:kr-main}
Let $G$ be a connected graph, $(T,\beta)$ be a greedy tree decomposition of $G$, and let $\tau \geq 2$ be such that $|\beta(t)| \leq \tau$ for every $t \in V(T)$.
Then $G$ contains a subcubic tree of treedepth $\Omega(\td(T) / \log \tau)$.
\end{lemma}

To this end, we first apply Lemma~\ref{lem:trees} to tree $T$ and obtain a subcubic tree $S$ such that
\begin{equation}\label{eq:treeS}
\td(S) \geq \td(T) \cdot \log_3(\varphi).
\end{equation}

Second, we apply the core part of the reasoning of Kawarabayashi and Rossman~\cite{KR18}.
The construction of Section~5 of~\cite{KR18} can be encapsulated in the following lemma.

\begin{lemma}[Section~5 of~\cite{KR18}]\label{lem:auxiliary}
Let $(T,\beta)$ be a greedy tree decomposition of graph $G$ and let $\tau = \max_{t \in V(T)} |\beta(t)|$.
Then for every subcubic subtree $S$ of $T$ there exists a subtree $F$ of $G$
such that $V(S) \subseteq V(F)$ and the maximum degree of $F$ is bounded by $\tau + 2$.
\end{lemma}

By application of Lemma~\ref{lem:auxiliary} to our decomposition $(T,\beta)$ and subtree $S$ we get a tree $F$ in $G$, which has large treedepth, as we show in a moment.
To this end, we need the following simple bound on treedepth of trees.
\begin{lemma}\label{lem:maxdeg}
For every tree $H$ with maximum degree bounded by $d \geq 2$ it holds that
\[ \log_d |V(H)| \leq \td(T) \leq 1+\log_2 |V(H)|. \]
\end{lemma}

\begin{proof}
We use the following equivalent recursive definition of treedepth: Treedepth of an empty graph is $0$, treedepth of a disconnected
graph equals the maximum of treedepth over its connected components, while for nonempty connected graphs $G$ we have $\td(G) = 1 + \min_{v \in V(G)} \td(G-v)$. 

For the lower bound, for $k \geq 1$ let $f_d(k)$ be the maximum possible number of vertices of a tree of maximum degree at most $d$
and treedepth at most $k$. Clearly, $f_d(1) = 1$. Since removing a single vertex from a tree of maximum degree at most $d$ results in
at most $d$ connected components, we have that
\[ f_d(k+1) \leq 1 + d \cdot f_d(k). \]
Consequently, we obtain by induction that
\[ f_d(k) \leq d^k - 1. \]
This proves the lower bound. For the upper bound, note that in every tree $T$ there exists a vertex $v \in V(T)$ such that every connected component of $T-\{v\}$ has at most $|V(T)|/2$ vertices.
Consequently, if we define $g(n)$ to be the maximum possible treedepth of an $n$-vertex tree, then $g(1) = 1$ and we have that
\[ g(n) \leq 1 + g(\lfloor n/2 \rfloor).\]
This proves the upper bound.
\end{proof}

By~\eqref{eq:treeS} and Lemma~\ref{lem:maxdeg} we get that $|V(S)| \geq 2^{\td(T) \cdot \log_3(\varphi)-1}$. This implies that also
\begin{equation}\label{eq:treeF}
|V(F)| \geq 2^{\td(T) \cdot \log_3(\varphi)-1}.
\end{equation}

As $S$ is subcubic, by Lemma~\ref{lem:auxiliary} we know that the maximum degree of $F$ is bounded by $\tau+2$.
Therefore Lemma~\ref{lem:maxdeg} and~\eqref{eq:treeF} jointly imply that
\begin{equation}\label{eq:twF}
\td(F) \geq \log_{(\tau + 2)} 2^{\td(T) \cdot \log_3(\varphi) - 1} \geq \frac{\td(T) \cdot \log_3(\varphi) - 1}{\log(\tau + 2)} = \Omega(\td(T) / \log \tau).
\end{equation}
Here, the last inequality follows from the assumption $\tau \geq 2$.

As tree $F$ is not necessarily subcubic, we apply one more time Lemma~\ref{lem:trees}
and get a subcubic subtree $H$ of $F$ such that
\begin{equation}
\td(H) \geq \td(F) \cdot \log_3(\varphi) = \Omega(\td(T) / \log \tau).
\end{equation}
which finishes the proof of Lemma~\ref{lem:kr-main}.

\medskip

With Lemma~\ref{lem:kr-main} in hand, we are ready to conlude the proof of Theorem~\ref{thm:apx}.
\begin{proof}[Proof of Theorem~\ref{thm:apx}.]
Without loss of generality we can assume that the input graph $G$ is connected.
As in the proof of Lemma~\ref{lem:td-cheap-apx}, we
apply the polynomial-time approximation algorithm for treewidth~\cite{FeigeHL08}, to compute
a tree decomposition $(T_0,\beta_0)$ of $G$ with $\Oh(n)$ nodes of $T_0$ and $|\beta(t)| \leq \tau$ for every $t \in V(T_0)$
and some $\tau = \Oh(\treewidth(G) \sqrt{\log \treewidth(G)})$.
As discussed in~\cite{KR18}, one can in polynomial time turn $(T_0,\beta_0)$ into a greedy tree decomposition
$(T,\beta)$ of $G$ without increasing the maximum size of a bag, that is, still $|\beta(t)| \leq \tau$
for every $t \in V(T)$.
We apply Lemma~\ref{lem:tw2td} to $(T,\beta)$, returning a treedepth decomposition of $G$ 
of width at most $\tau \cdot \td(T) = \Oh(\td(T) \tw(G) \sqrt{\log \tw(G)})$.

It remains to bound $\td(T)$. Lemma~\ref{lem:kr-main} asserts that $G$ contains a subcubic tree $H$
of treedepth $\Omega(\td(T) / \log \tau)$. 
Therefore $\td(T) = \Oh(\td(H) \log \tau) = \Oh(\td(G) \log \tw(G))$
and thus the width of the computed treedepth decomposition
is $\Oh(\td(G) \tw(G) \log^{3/2} \tw(G))$. This finishes the proof of Theorem~\ref{thm:apx}.
\end{proof}

\section{Proof of Theorem~\ref{thm:lincol}}\label{sec:lincol}
Here we show how to assemble the proof of Theorem~\ref{thm:lincol}
from Theorem~\ref{thm:main}, a number of intermediate results of~\cite{LinCol},
and an improved excluded grid theorem due to Chuzhoy and Tan~\cite{ChuzhoyT19}:
\begin{theorem}[\cite{ChuzhoyT19}]\label{thm:gmt}
There exists a polynomial $p_\mathrm{GMT}$ such that for every integer $k$
if a graph $G$ has treewidth at least $k^9 p_\mathrm{GMT}(\log k)$ then $G$ contains a $k \times k$ grid as a minor.
\end{theorem}
The following two results were proven in~\cite{LinCol}.
\begin{lemma}[\cite{LinCol}]\label{lem:grid-lincol}
If a graph $G$ contains a $k \times k$ grid as a minor, then every linear coloring of $G$
requires $\Omega(\sqrt{k})$ colors.
\end{lemma}
\begin{lemma}[\cite{LinCol}]\label{lem:tree-lincol}
If $G$ is a tree of treedepth $d$ and maximum degree $\Delta$, then every linear
coloring of $G$ requires at least $d / \log_2(\Delta)$ colors.
\end{lemma}
Recall that Theorem~\ref{thm:main} asserts that there exists a constant $C$ such that for every graph $G$
and integers $a,b \geq 2$, if the treedepth of $G$ is at least $Cab \log a$, then
either the treewidth of $G$ is at least $a$ or $G$ contains a subcubic tree of treedepth at least $b$.
Applying this theorem to $a = \theta(k^2)$ and $b = k \log_2(3)$, one obtains that 
if the treedepth of $G$ is $\Omega(k^{19} p_\mathrm{GMT}(\log k) \log k)$, then 
$G$ contains either a $\theta(k^2) \times \theta(k^2)$ grid minor or a subcubic tree
of treedepth at least $k \log_2(3)$. In the first outcome, Lemma~\ref{lem:grid-lincol}
gives the desired number of colors of a linear coloring, while in the second outcome
the same result is obtained from Lemma~\ref{lem:tree-lincol}. This concludes the proof of Theorem~\ref{thm:lincol}.

\section{An example of a tree with treedepth quadratic in the height of
  the binary tree or logarithm of a length of a path}\label{sec:ex}
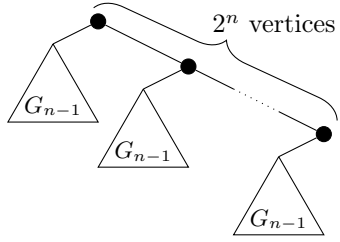
\begin{figure}[tb]
\begin{center}
\begin{tikzpicture}[scale=0.6]
   \tikzstyle{vertex}=[circle,draw=black,fill=black,minimum size=0.2cm,inner sep=0pt]

   \node[vertex] (x1) at (0, 0) {};
   \node[vertex] (x2) at (2, -1) {};
   \node[vertex] (x3) at (5, -2.5) {};
   \draw (x1) -- (x2) -- (3, -1.5);
   \draw[dotted] (3, -1.5) -- (4, -2);
   \draw (4, -2) -- (x3);

   \draw[decorate,decoration={brace,amplitude=10pt}] ($ (x1) + (-0.1, 0.3) $) -- ($ (x3) + (0.3, 0.3) $);
   \draw ($ (2.5, -1.25) + (1.4, 1.2) $) node {$2^n$ vertices};

   \foreach \n in {1, 2, 3} {
     \coordinate (r) at ($ (x\n) + (-1, -0.5) $);
     \draw (r) -- ($ (r) + (-1, -1.732) $) -- ($ (r) + (1, -1.732) $) -- (r);
     \draw ($ (r) + (0, -1.4) $) node {\small $G_{n-1}$};
     \draw (r) -- (x\n);
   }

\end{tikzpicture}
\caption{Construction of $G_n$.}\label{fig:Gn}
\end{center}
\end{figure}
In this section we provide a construction of a family of trees $(G_n)_{n \geq 1}$ such that
\begin{enumerate}
\item The tree $G_n$ does not contain a path with $2^{n+2}$ vertices.\label{p:ex:no-path}
\item The tree $G_n$ does not contain a subdivision of a full binary tree of depth $n+2$.\label{p:ex:no-tree}
\item The treedepth of $G_n$ is at least $\binom{n+1}{2}$.\label{p:ex:td}
\end{enumerate}

We will consider each tree $G_n$ as a rooted tree. The tree $G_1$ consists of a single vertex. 
For $n \geq 2$, the tree $G_n$ is defined recursively as follows.
We take a path $P_n$ with $2^n$ vertices and for each $v \in V(P_n)$ we create a copy $C_n^v$ of $G_{n-1}$ and attach
its root to $v$. We root $G_n$ in one of the endpoints of $P_n$; see Figure~\ref{fig:Gn}.
We now proceed with the proof of the properties of $G_n$.

Since every path in $G_n$ is contained in at most two root-to-leaf paths (not necessarily edge-disjoint), to show Property~\eqref{p:ex:no-path}
it suffices to show the following.
\begin{lemma}\label{lem:ex:no-path}
Every root-to-leaf path in $G_n$ contains less than $2^{n+1}$ vertices.
\end{lemma}
\begin{proof}
We prove the statement by induction on $n$.
For $n =1$ the statement is straightforward.
For the inductive step, observe that every root-to-leaf path in $G_n$ consists of a subpath of $P_n$
(which has $2^n$ vertices) and a root-to-leaf path in one of the copies $C_n^v$ of $G_{n-1}$ (which has
    less than $2^n$ vertices by the inductive assumption).
\end{proof}

We say that a subtree $H$ of $G_n$ that is a subdivision of a full binary tree of height $h \geq 1$ is \emph{aligned}
if $h=1$ or $h \geq 2$ and the closest to the root vertex of $H$ is of degree $2$ in $H$ and its deletion 
breaks $H$ into two subtrees containing a subdivision of a full binary tree of height $h-1$.
In other words, an aligned subtree has the same ancestor-descendant relation as the tree $G_n$.
Observe that any subtree $H_0$ of $G_n$ that is a subdivision of a full binary tree of height $h \geq 2$
contains a subtree that is an aligned subdivision of a full binary tree of height $h-1$.
Therefore, to prove Property~\eqref{p:ex:no-tree}, it suffices to show the following.
\begin{lemma}\label{lem:ex:no-tree}
$G_n$ does not contain an aligned subdivision of a full binary tree of height $n+1$.
\end{lemma}
\begin{proof}
We prove the claim by induction on $n$. It is straightforward for $n=1$.
For $n \geq 2$, let $H$ be such an aligned subtree of $G_n$ and let $w$ be the closest to the root of $G_n$ vertex
of $H$. If $w \in V(C_n^v)$ for some $v \in V(P_n)$, then $H$ is completely contained in $C_n^v$, which is a copy
of $G_{n-1}$. Otherwise, $w \in V(P_n)$ and thus one of the components of $H-\{w\}$ lies in
$C_n^w$. However, this component contains an aligned subdivision of a full binary tree of height $n$.
In both cases, we obtain a contradiction with the inductive assumption.
\end{proof}

We are left with the treedepth lower bound of Property~\eqref{p:ex:td}.
To this end, we consider the following families of trees. For integers $a,b \geq 1$,
the family $\mathcal{G}_{a,b}$ contains all trees $H$ that are constructed from a path 
$P_H$ with at least $2^a$ vertices by attaching, for every $v \in V(P_H)$, a tree $T_v$ of treedepth at least $b$
by an edge to $v$.
We show the following.
\begin{lemma}\label{lem:ex:td}
For every $H \in \mathcal{G}_{a,b}$ we have $\td(H) \geq a + b$.
\end{lemma}
\begin{proof}
We prove the lemma by induction on $a$. 
For $a=1$ we have $\td(H) \geq a + 1$ and $H$ contains two vertex-disjoint subtrees of treedepth at least $b$ each.
Assume then $a > 1$ and $H \in \mathcal{G}_{a,b}$. 
Then for every $v \in V(H)$, $H-v$ contains a connected component that contains a subtree belonging to $\mathcal{G}_{a-1,b}$.
This finishes the proof.
\end{proof}
We show Property~\eqref{p:ex:td} by induction on $n$. Clearly, $\td(G_1) = 1 = \binom{1+1}{2}$.
Consider $n \geq 2$. Since the treedepth of $G_{n-1}$ is at least $\binom{n}{2}$, we have
that $G_n \in \mathcal{G}_{n, \binom{n}{2}}$. By Lemma~\ref{lem:ex:td},
we have that
\[ \td(G_n) \geq n + \binom{n}{2} = \binom{n+1}{2}. \]
This finishes the proof of Property~\eqref{p:ex:td}.

\section{Conclusions}
We have provided improved bounds for the excluded minor approximation of treedepth of Kawarabayashi and Rossman~\cite{KR18}.
Our main result, Theorem~\ref{thm:main}, is close to being optimal in the following sense:
as witnessed by the family of trees, if one considers the measure $r := \td(G) / \tw(G)$, 
 one cannot hope to find a tree in $G$ of treedepth larger than $r$. 
Improving the $Ck^3$ bound of Corollary~\ref{cor:main} to $Ck^{3-\varepsilon}$ for some $\varepsilon > 0$ seems challenging.

Our techniques can be applied to a polynomial-time treedepth approximation algorithm, improving upon
state-of-the-art tradeoff trick. As a second open problem, we ask for
a polynomial-time or single-exponential in treedepth
parameterized algorithm for constant or polylogarithmic approximation of treedepth.

\bibliographystyle{abbrv}

\bibliography{references}

\end{document}